\newcommand{\fullVersion}{}

\date{}

\ifdefined\lipics{}
\documentclass[a4paper,UKenglish]{lipics-v2018}
\else
\documentclass[11pt]{article}

\usepackage[margin=1in]{geometry}
\usepackage{amsmath}
\usepackage{amsthm}
\usepackage{bm}
\newtheorem{theorem}{Theorem}[section]
\newtheorem{lemma}[theorem]{Lemma}
\newtheorem{corollary}[theorem]{Corollary}
\newtheorem{claim}[theorem]{Claim}
\newtheorem{proposition}[theorem]{Proposition}
\newtheorem{definition}[theorem]{Definition}

\newcommand{\ILPcomplexity}{\ensuremath{O\parentheses{(1+f/\log n)\cdot\parentheses{\frac{\log \Delta}{ \log \log \Delta} + \parentheses{f\cdot\log M}^{1.01}\cdot \log\eps^{-1}\cdot (\log\Delta)^{0.01}}}}}
\usepackage{amsfonts}
\usepackage{amssymb}
\usepackage{url}
\fi
\usepackage[utf8]{inputenc}
\usepackage[T1]{fontenc}
\usepackage{lmodern}

\usepackage{graphicx,comment}
\usepackage[linesnumbered,noresetcount,vlined]{algorithm2e}
\usepackage{environ,tabularx}

\usepackage{xcolor}
\usepackage[inline]{enumitem}
\usepackage{mathrsfs}

\newcommand{\xmax}{M}

\newenvironment{lemma-repeat}[1]{\begin{trivlist}
\item[\hspace{\labelsep}{\bf\noindent Lemma \ref{#1} }]\em }%
{\end{trivlist}}
\newenvironment{theorem-repeat}[1]{\begin{trivlist}
\item[\hspace{\labelsep}{\bf\noindent Theorem \ref{#1} }]\em }%
{\end{trivlist}}

\newcommand*\samethanks[1][\value{footnote}]{\footnotemark[#1]}
\DeclareMathAlphabet{\mathpzc}{OT1}{pzc}{m}{it}
\ifdefined\noComments
\newcommand{\Gcomment}[1]{}
\newcommand{\Ecomment}[1]{}
\newcommand{\Rcomment}[1]{}
\else
\newcommand{\Gcomment}[1]{{\color{red} Gregory: #1}}
\newcommand{\Ecomment}[1]{{\color{red} Guy: #1}}
\newcommand{\Rcomment}[1]{{\color{green} Ran: #1}}

\fi

\newcommand{\ceil}[1]{\lceil #1 \rceil}
\newcommand{\floor}[1]{\left\lfloor #1 \right\rfloor}
\renewcommand{\epsilon}{\varepsilon}
\newcommand{\eps}{\varepsilon}
\newcommand{\RR}{\mathbb{R}}
\newcommand{\NN}{\mathbb{N}}
\newcommand{\congest}{\textsc{congest}\xspace}
\newcommand{\local}{\textsc{local}\xspace}
\newcommand{\WHVC}{\textsc{mwhvc}\xspace}

\newcommand{\MWHVC}{\textsc{mwhvc}\xspace}
\newcommand{\MWVC}{\textsc{mwvc}\xspace}
\newcommand{\Dv}{\ensuremath{\Delta}}
\newcommand{\opt}{\textsf{opt}}
\DeclareMathOperator*{\argmin}{argmin} 
\DeclareMathOperator*{\alg}{\textup{\MWHVC\xspace}} 
\DeclareMathOperator*{\deal}{deal} 
\newcommand{\parentheses}[1]{\left(#1\right)}
\newcommand{\logp}[1]{\log\parentheses{#1}}

\newcommand{\set}[1]{\left\{ #1 \right\}}

\makeatletter
\newcommand{\problemtitle}[1]{\gdef\@problemtitle{#1}}
\newcommand{\probleminput}[1]{\gdef\@probleminput{#1}}
\newcommand{\problemoutput}[1]{\gdef\@problemoutput{#1}}
\newcommand{\problemobj}[1]{\gdef\@problemobj{#1}}
\NewEnviron{problem}{
  \problemtitle{}\probleminput{}\problemoutput{}\problemobj{}
  \BODY
  \par\addvspace{.5\baselineskip}
  \noindent
  \begin{tabularx}{\textwidth}{@{\hspace{\parindent}} l X c}
    \textbf{Input:} & \@probleminput \\
    \textbf{Output:} & \@problemoutput \\
    \textbf{Objective:} & \@problemobj
  \end{tabularx}
  \par\addvspace{.5\baselineskip}
}
\makeatother

\usepackage{microtype}

\usepackage{titlesec}
\titlespacing{\section}{1.5pt}{*1.5}{*1.5}
\titlespacing{\subsection}{1pt}{*1.5}{*1.5}
\titlespacing{\subsubsection}{1pt}{*1}{*1}
\titlespacing{\paragraph}{1pt}{0.5pt}{0.5pt}[]

\ifdefined\lipics{}
\titlerunning{$(f+\epsilon)$-Approximation for Weighted Hypergraph Vertex Cover in $O\left(\log{\Delta}\right)$ Rounds}
\fi

\author{Ran Ben-Basat\thanks{Harvard University,  \texttt{ran@seas.harvard.edu}}
\and Guy Even\thanks{Tel Aviv University, \texttt{guy@eng.tau.ac.il}}
			\and Ken-ichi Kawarabayashi \thanks{NII, Japan, \texttt{k\_keniti@nii.ac.jp}, \texttt{greg@nii.ac.jp}} \and Gregory
			Schwartzman\samethanks}
\ifdefined\lipics{}
\subjclass{Theory of computation$\rightarrow$Distributed algorithms}


\category{}

\EventEditors{}
\EventNoEds{4}
\EventLongTitle{}
\EventShortTitle{DISC 2018}
\EventAcronym{DISC}
\EventYear{2018}
\EventDate{}
\EventLocation{}
\EventLogo{eatcs}
\SeriesVolume{}
\ArticleNo{}
\fi

\begin{document}
\begin{titlepage}
\title{Optimal Distributed Covering Algorithms}
\maketitle

\begin{abstract}

We present a time-optimal deterministic distributed algorithm for approximating a minimum weight vertex cover in hypergraphs of rank $f$.
This problem is equivalent to the Minimum Weight Set Cover problem in which the frequency of every element is bounded by $f$.
The approximation factor of our algorithm is $(f+\epsilon)$.
Let $\Delta$ denote the maximum degree in the hypergraph.
Our algorithm runs in the \congest model and requires $O(\log{\Delta} / \log \log \Delta)$ rounds, for constants $\epsilon \in (0,1]$ and $f\in\mathbb N^+$.
This is the first distributed algorithm for this problem whose running time does not depend on the vertex weights nor the number of vertices. Thus adding another member to the exclusive family of \emph{provably optimal} distributed algorithms.

For constant values of $f$ and $\epsilon$, our algorithm improves over the $(f+\epsilon)$-approximation algorithm of \cite{KuhnMW06} whose running time is $O(\log \Delta + \log W)$, where $W$ is the ratio between the largest and smallest vertex weights in the graph. Our algorithm also achieves an $f$-approximation for the problem in $O(f\log n)$ rounds, improving over the classical result of \cite{KhullerVY94} that achieves a running time of $O(f\log^2 n)$. Finally, for weighted vertex cover ($f=2$) our algorithm achieves a \emph{deterministic} running time of $O(\log n)$, matching the \emph{randomized} previously best result of~\cite{KoufogiannakisY2011}.

We also show that integer covering-programs can be reduced to the Minimum Weight Set Cover problem in the distributed setting. This allows us to achieve an $(f+\epsilon)$-approximate integral solution in \ILPcomplexity{} rounds, where $f$ bounds the number of variables in a constraint, $\Delta$ bounds the number of constraints a variable appears in, and $\xmax=\max \set{1, \ceil{1/a_{\min}}}$, where $a_{\min}$ is the smallest normalized constraint coefficient. This 
improves over the results of \cite{KuhnMW06} for the integral case, which combined with 
rounding achieves the same guarantees in $O\parentheses{\epsilon^{-4}\cdot f^4\cdot \log f\cdot\log(M\cdot\Delta)}$ rounds. 


\end{abstract}
\thispagestyle{empty}
\end{titlepage}

\section{Introduction}
In the Minimum Weight Hypergraph Vertex Cover (\MWHVC) problem, we are given a hypergraph $G=(V,E)$ with vertex weights $w:V\to \set{1,\ldots,W}$.\footnote{Let $n\triangleq |V|$. We assume that $|E|=n^{O(1)}$ and $W=n^{O(1)}$.} The goal is to find a minimum weight \emph{cover} $U\subseteq V$ such that $\forall e\in E: e\cap U\neq\emptyset$. In this paper we develop a distributed approximation algorithm for \MWHVC in the \congest model. The approximation ratio is $f+\eps$, where $f$ denotes the rank of the hypergraph (i.e., $f$ is an upper on the size of every hyperedge). 
The \MWHVC problem is a generalization of the Minimum Weight Vertex Cover (\MWVC) problem (in which $f=2$). The \MWHVC problem is also equivalent to the Minimum Weight Set Cover Problem (the rank $f$ of the hypergraph corresponds to the maximum frequency of an element). Both  of these problems are among the classical NP-hard problems presented in~\cite{Karp72}. 

We consider the following distributed setting for 
the \MWHVC problem. The communication network is a bipartite graph $H(E\cup V,\set{\set{e,v} \mid v\in e})$. We refer to the network vertices as \emph{nodes} and network edges as \emph{links}. The nodes
of the network are the hypergraph vertices on one side and hyperedges on the other side. There is a network link between vertex $v\in V$ and hyperedge $e\in E$ iff $v\in e$. The computation is performed in synchronous rounds, where messages are sent between neighbors in the communication network. As for message size, we consider the \congest model where message sizes are bounded to $O(\log|V|)$. This is more
restrictive than the \local model \mbox{where message sizes are unbounded.}

\subsection{Related work}
We survey previous results for \MWHVC and \MWVC. A comprehensive list of previous results appears in Tables~\ref{VCtbl} and~Table~\ref{tbl}. 
\paragraph{Vertex Cover.} \quad{}
The understanding of the round complexity for distributed \MWVC has been established in two papers: a lower bound in~\cite{KuhnMW16} and
a matching upper bound in~\cite{Bar-YehudaCS17}.
Let $\Delta$ denote the maximum vertex degree in the graph $G$.
The lower bound states that any distributed constant-factor approximation algorithm requires $\Omega(\log\Delta / \log\log \Delta)$ rounds to terminate. This lower bound holds for every constant approximation ratio, for unweighted graphs and even if the message lengths are not bounded (i.e., LOCAL model)~\cite{KuhnMW16}. 
The matching upper bound is a 
$(2+\epsilon)$-approximation distributed algorithm in the \congest model, for every $\epsilon=\Omega(\log\log\Delta / \log \Delta)$.\footnote{Recently,   the range of $\epsilon$ for which the runtime is optimal was improved to $\Omega(\log^{-c}\Delta)$ for any  $c=O(1)$~\cite{BEKS18}.}
In \cite{KoufogiannakisY2011} an $O(\log n)$-round 2-approximation randomized algorithm for weighted graphs in the \congest model is given. We note that \cite{KoufogiannakisY2011} was the first to achieve this running time with no dependence on $W$, the maximum weight of the nodes.

\paragraph{Hypergraph Vertex Cover.} \quad{} For constant values of $f$, Astrand et al.~\cite{AstrandS10} present an $f$-approximation algorithm for anonymous networks whose running time is $O(\Delta^2 + \Delta\cdot\log^* W)$.
Khuller et al.~\cite{KhullerVY94} provide a solution that runs in $O(f\log 1/ \epsilon \cdot \log n)$ rounds in the \congest model for any $\epsilon > 0$ and achieves an $(f+\epsilon)$-approximation. Setting $\epsilon = 1/W$ (recall that $W=poly(n)$) results in a $f$-approximation in $O(f\log^2 n)$-rounds. For constant $\eps$ and $f$ values, Kuhn et al.~\cite{Kuhn2005,KuhnMW06} present an $(f+\epsilon)$-approximation algorithm that terminates in $O(\log \Delta + \log W)$ rounds. 

For the Minimum Cardinality Vertex Cover in Hypergraphs Problem, the lower bound was recently matched by~\cite{unweightedHVC} with an $(f+\eps)$-approximation algorithm in the \congest model.  The round complexity in~\cite{unweightedHVC} is
$O\parentheses{f/\epsilon \cdot \frac{\log(f\cdot\Delta)}{\log\log(f\cdot\Delta)}}$, which is optimal for constant $f$ and $\eps$. 
The algorithm in~\cite{unweightedHVC} and its analysis is a deterministic version of the randomized maximal independent set algorithm
of~\cite{ghaffari2016improved}.

\subsection{Our contributions}

In this paper, we present a deterministic distributed $(f+\epsilon)$-approximation algorithm for minimum weight vertex cover
in $f$-rank hypergraphs, which completes in
$$
O\parentheses{f\cdot \log(f/\eps) + \frac{\log \Delta}{ \log \log \Delta}+\min\set{\log\Delta,f\cdot \log(f/\eps)\cdot (\log\Delta)^{0.001}}}
  $$
rounds in the \congest model. 
For any constants $\epsilon \in (0,1)$ and $f\in\mathbb N^+$ this implies a running time of $O(\log{\Delta} / \log \log \Delta)$, which is optimal according to \cite{KuhnMW16}.
This is the first distributed algorithm for this problem whose round complexity does not depend on the node weights nor the number of vertices. 

Our algorithm is one of a handful of distributed algorithms for \emph{local} problems which are \emph{provably optimal} \cite{ColeV86, ChangKP16, Bar-YehudaCGS17, Bar-YehudaCS17, GhaffariS17, unweightedHVC}. Among these are the classic Cole-Vishkin algorithm \cite{ColeV86} for 3-coloring a ring, the more recent results of \cite{Bar-YehudaCGS17} and \cite{Bar-YehudaCS17} for \MWVC and Maximum Matching, and the result of \cite{unweightedHVC} for Minimum Cardinality Hypergraph Vertex Cover.

Our algorithm also achieves a \emph{deterministic} $f$-approximation for the problem in $O(f\log n)$ rounds. This improves over the best known result for hypergraphs $O(f\log^2 n)$ \cite{KhullerVY94} and matches the best known \emph{randomized} results for weighted vertex cover ($f=2$) of $O(\log n)$-rounds \cite{KoufogiannakisY2011}.

We also show that general covering Integer Linear Programs (ILPs) can be reduced to 
\MWHVC
in the distributed setting. That is,  LP constraints can be translated into hyperedges such that a cover for the hyperedges satisfies all covering constraints.
This allows us to achieve an $(f+\epsilon)$-approximate \emph{integral} solution in $\ILPcomplexity$ rounds, where $f$ bounds the number of variables in a constraint, $\Delta$ bounds the number of constraints a variable appears in, and $\xmax=\max \set{1, \ceil{1/a_{\min}}}$, where $a_{min}$ smallest normalized constraint coefficient. This significantly improves over the results of \cite{KuhnMW06} for the integral case, which combined with 
rounding achieves the same guarantees in $O\parentheses{\epsilon^{-4}\cdot f^4\cdot \log f\cdot\log(M\cdot\Delta)}$ rounds. Note that the results of \cite{KuhnMW06} also include a $(1+\epsilon)$-approximation for the fractional case, while our result only allows for an integral solution. We also note that plugging $\epsilon=1/nWM$ into our algorithm, achieves an $f$-approximation for ILPs in polylogarithmic time, a similar result cannot be achieved using \cite{KuhnMW06}.
 \begin{table*}[t!]
	\centering{
		\resizebox{1.0 \textwidth}{!}{
			\begin{tabular}{|cclll|}
				\hline
				det.& weighted& approximation& time & algorithm\\
				\hline

				yes& no& 3&$O(\Delta)$& \cite{PolishchukS09}\\

				
				yes& no& 2&$O(\Delta^{2})$& \cite{AstrandFPRSU09}\\
				
				yes& yes& 2&$O(1)$ for $\Delta\le 3$& \cite{AstrandFPRSU09}\\
				
				
				yes& yes& 2&$O(\Delta+\log^{*}n)$& \cite{PanconesiR01}\\
				
				
				yes& yes& 2&$O(\Delta + \log^* W)$& \cite{AstrandS10}\\
				
				
				yes& yes& $ 2$&$O(\log^2 n)$&\cite{KhullerVY94}\\
				
				yes& yes& 2&$O(\log n\log{\Delta}/\log^2\log{\Delta})$& \cite{BEKS18}\\
				
				no& yes& 2&$O(\log n)$& \cite{GrandoniKP08,KoufogiannakisY2011}\\
				
				\textbf{yes}& \textbf{yes}& \textbf{2}&$\bm{O(\log n)}$& \textbf{This work}\\
				\hline
				
				yes& yes& $ 2+\epsilon$&$O(\epsilon^{-4}\log(W\cdot\Delta))$& \cite{Hochbaum82, KuhnMW06}\\
				
				yes& yes& $ 2+\epsilon$&$O(\log\epsilon^{-1}\log n)$& \cite{KhullerVY94}\\
				
				yes& yes& $ 2+\epsilon$&$O(\epsilon^{-1}\log\Delta/\log\log\Delta)$& \cite{Bar-YehudaCS17,unweightedHVC}\\	
				
				yes& yes& ${2+\epsilon}$&${O\parentheses{\frac{\log \Dv}{\log\log \Dv} + {\frac{\log\epsilon^{-1}\log \Dv}{\log^2 \log \Dv}}}}$& \cite{BEKS18}\\	
				
				
				\textbf{yes}& \textbf{yes}& $\bm{2+\epsilon}$&$\bm{O\parentheses{\frac{\log \Delta}{ \log \log \Delta} + \log{}\eps^{-1}\cdot (\log{} \Delta)^{0.001}}}$& \textbf{This work}\\	
				\hline
				
				yes& yes& $ 2+\frac{\log\log\Delta}{c\cdot\log\Delta}$&$O(\log\Delta/\log\log\Delta)$& \cite{Bar-YehudaCS17}, $\forall c=O(1)$\\								

				yes& yes& ${{2+\parentheses{\log\Delta}^{-c}}}$&$O(\log\Delta/\log\log\Delta)$& \cite{BEKS18}, $\forall c=O(1)$\\								
				
				\textbf{yes}& \textbf{yes}& $\bm{2+2^{-c\cdot\parentheses{\log\Delta}^{0.99}}}$
				&$\bm{O(\log{}\Delta/{\log\log\Delta})}$& \textbf{This work, $\bm{\forall c=O(1)}$}\\	

				\hline
			\end{tabular}
		}
	}
	\caption{Previous distributed algorithms for \MWVC.
		In the table, $n=|V|$ and $\epsilon\in(0,1)$. 
		Some of the algorithms hold only for the unweighted case and some are randomized. For randomized algorithms the running times hold in expectation or with high~probability.
	}
	\label{VCtbl}
\end{table*}

 \begin{table*}[]
	\centering{
		\resizebox{1.0 \textwidth}{!}{
			\begin{tabular}{|clll|}
				\hline
				weighted& approximation& time & algorithm\\
				\hline
				
				yes& $f$&$O\parentheses{f^2\Delta^2+f\Delta\log^* W}$ & \cite{AstrandS10}\\
				
				yes& $f$&$O\parentheses{f \log^2 n}$ & \cite{KhullerVY94}\\
				
				\textbf{yes}& $\bm{f}$&$\bm{O\parentheses{f \log n}}$& \textbf{This work}\\
				\hline
				
				no& $ f+\epsilon$&$O\parentheses{\epsilon^{-1}\cdot f\cdot \frac{\log(f\Delta)}{\log\log(f\Delta)}}$& \cite{unweightedHVC}\footnotemark\\	
				
				yes& $f+\epsilon$&$O\parentheses{f\cdot \log(f/\eps)\cdot \log n}$ & \cite{KhullerVY94}\\
				
				yes& $ f+\epsilon$&$O\parentheses{\epsilon^{-4}\cdot f^4\cdot \log f\cdot\log(W\cdot\Delta)}$& \cite{KuhnMW06}\\
				
				
				\textbf{yes}& $\bm{f+\epsilon}$&$\bm{O\parentheses{f\cdot\log(f/\eps)\cdot (\log{}\Delta)^{0.001}+\frac{\log \Delta}{ \log \log \Delta}}}$& \textbf{This work}\\	
				\hline		
				
				%
				
				no& $ f+1/c$&$O\parentheses{\log\Delta/{\log\log\Delta}}$& \cite{unweightedHVC}, $\forall f,c=O(1)$\\
				
				\textbf{yes}& $\bm{f+2^{-c\cdot\parentheses{\log\Delta}^{0.99}}}$
				&$\bm{O\parentheses{\log{}\Delta/{\log\log\Delta}}}$& \textbf{This work}, $\bm{\forall f,c=O(1)}$\\
				
				
				\hline
			\end{tabular}
		}
	}
	\caption{Previous distributed algorithms for \MWHVC.
		In the table, $n=|V|$ and $\epsilon\in(0,1)$. 
		All algorithms are deterministic. Note that \cite{unweightedHVC} holds only for unweighted hypergraphs.
	}
	\label{tbl}\vspace{-4mm}
\end{table*}
\footnotetext{The authors state their result for an $f(1+\eps)$-approximation which removes the $f$ factor from the runtime.}
\subsection{Tools and techniques} 
\paragraph {The Primal-Dual schema.} \quad{}The Primal-Dual approach introduces, for every hyperedge $e \in E$, a dual variable denoted by $\delta(e)$. The dual edge packing constraints are $\forall v\in V,\sum_{v \in e} \delta(e) \leq w(v)$. If for some $\beta \in [0,1)$ it holds that $\sum_{v \in e} \delta(e) \geq (1-\beta)\cdot w(v)$, we say the $v$ is $\beta$-tight. Let $\beta=\eps/(f+\eps)$. For every feasible dual solution, the weight of the set of $\beta$-tight vertices is at most $(f+\eps)$ times the weight of an optimal (fractional) solution. The algorithm terminates when the set of $\beta$-tight edges constitutes a vertex cover.

\paragraph{The challenge.} \quad{} When designing a Primal-Dual distributed algorithms, the main challenge is in controlling the rate at which we increase the dual variables. On the one hand, we must grow them rapidly to reduce the number of communication rounds. On the other hand, we may not violate the edge packing constraints. This is tricky in the distributed environments as we have to coordinate between nodes. For example, the result of \cite{Bar-YehudaCS17} does not generalize to hypergraphs, as hyperedges require the coordination of more than two nodes in order to increment edge variables. 

\paragraph{Our algorithm.} \quad{}The algorithm proceeds in iterations, each of which requires a constant number of communication rounds. We initialize the dual variables in a "safe" way so that feasibility is guaranteed. We refer to the additive increase of the dual variable $\delta(e)$ as $\deal(e)$. 
Similarly to~\cite{BEKS18}, we use \emph{levels} to measure the progress made by a vertex. Whenever the level of a vertex increases, it sends a message about it to all incident edges, which multiply (decrease) their deals by $0.5$. 
Intuitively, the level of a vertex equals the logarithm of its uncovered portion.
Formally, we define the level of a vertex $v$ as $\ell(v)\triangleq \floor{\log \frac{w(v)}{w(v)-\sum_{e\ni v} \delta(e)}}$. That is, the initial level of $v$ is $0$
and it is increased as the dual variables of the edges incident to $v$ grow. 
The level of a vertex never reaches $z\triangleq \ceil{\log\beta^{-1}}$ as this implies that it is $\beta$-tight and entered the cover.
Loosely speaking, the algorithm  increases the increments $\deal(e)$ exponentially (multiplication by $\alpha$) provided that no vertex $v\in e$ is $(0.5^{\ell(v)} / \alpha)$-tight with respect to the deals of the previous iteration. 
Here, $\alpha\ge 2$ is a positive parameter that we determine later. 
The analysis builds on two observations: (1)~The number of times that the increment $\deal(e)$ is multiplied by $\alpha$ is bounded by $\log_\alpha \Delta$.
(2)~The number of iterations in which $\deal(e)$ is not multiplied by $\alpha$ is bounded by $O(f\cdot z\cdot \alpha)$. Loosely speaking, each such iteration means that for at least one vertex $v\in e$ the sum of deals is at least an $1/(2\alpha)$-fraction of its slack. Therefore, after at most $O(\alpha)$ such iterations that vertex will level up. Since there are $z$ levels per vertex and $f$ vertices in $e$, we have that the number of such iterations is at most $O(f\cdot z\cdot \alpha)$.
%
Hence the total number of iterations is bounded by $O(\log_\alpha \Delta+f\cdot z\cdot \alpha)$. 

\paragraph{Integer linear programs (ILPs).}\quad
We show distributed reductions that allow us to compute an $(f+\eps)$-approximation for general covering integer linear programs (ILPs). To that end, we first show that any Zero-One covering program (where all variables are binary) can be translated into a set cover instance in which the vertex degree is bounded by $2^f$ times the bound on the number of constraints each ILP variable appears in. We then generalize to arbitrary covering ILPs by replacing each variable with multiple vertices in the hypergraph, such that the value associated with the ILP variable will be the weighted sum of the vertices in the cover.

\section{Problem Formulation}
Let $G=(V,E)$ denote a hypergraph. Vertices in $V$ are equipped with positive weights $w(v)$.  For a subset $U\subseteq V$, let $w(U)\triangleq \sum_{v\in U} w(v)$.  Let $E(U)$ denote the set of hyperedges that are incident to some vertex in $U$ (i.e., $E(U)\triangleq \{ e\in E \mid e\cap U \neq \emptyset\}$).

\medskip
\noindent
The \emph{Minimum Weight Hypergraph Vertex Cover} Problem (\WHVC) is
defined as follows.
\begin{problem}
  \probleminput{Hypergraph $G=(V,E)$ with vertex weights $w(v)$.}%
  \problemoutput{A subset $C\subseteq V$ such that $E(C)=E$. }%
  \problemobj{Minimize $w(C)$.}
\end{problem}
 
The \WHVC Problem is equivalent to the Weighted Set Cover Problem.
Consider a set system $(X,\mathcal{U})$, where $X$ denotes a set of elements and $\mathcal{U}=\set{U_1,\ldots, U_m}$ denotes a collection of subsets of $X$. 
The reduction from the set system $(X,\mathcal{U})$ to a hypergraph $G=(V,E)$ proceeds as follows.
The set of vertices is $V\triangleq \set{u_1,\ldots, u_m}$ (one vertex $u_i$ per subset $U_i$).
The set of edges is $E\triangleq\set{e_x}_{x\in X}$ (one hyperedge $e_x$ per element $x$), where $e_x\triangleq \set{u_i: x\in U_i}$.
The weight of vertex $u_i$ equals the weight of the subset $U_i$.

\section{Distributed 
Approximation Algorithm for {\small MWHVC}}
\label{sec: alg}
\subsection{Input}
The input is a hypergraph $G=(V,E)$ with non-negative vertex weights $w:V\rightarrow \mathbb{N}^+$ and an approximation ratio parameter $\eps\in (0,1]$.  
We denote the number of vertices by $n$, the rank of $G$ by $f$ (i.e., each hyperedge contains at most $f$ vertices), and the maximum degree of $G$ by $\Delta$ (i.e., each vertex belongs to at most $\Delta$ edges).

\paragraph{Assumptions.} \quad{}
We assume that 
\begin{enumerate*}[label={(\roman*)}]
\item Vertex weights are polynomial in $n\triangleq|V|$ so that sending a vertex weight requires $O(\log n)$ bits. 
\item Vertex degrees are polynomial in $n$ (i.e., $|E(v)|=n^{O(1)}$) so that sending a vertex degree requires $O(\log n)$ bits. 
Since $|E(v)|\leq n^f$, this assumption trivially holds for constant $f$.
\item The maximum degree is at least $3$ so that $\log \log \Delta >0$.
\end{enumerate*}

\subsection{Output}
A vertex cover $C\subseteq V$. Namely, for every hyperedge $e\in E$, the intersection $e\cap C$ is not empty.  The set $C$ is maintained locally in the sense that every vertex $v$ knows whether it \mbox{belongs to $C$ or not.}

\subsection{Communication Network} 
The communication network $N(E\cup V,\set{\set{e,v} \mid v\in e})$ is a bipartite graph.  There are two types of nodes in the network: \emph{servers} and \emph{clients}. The set of servers is $V$ (the vertex set of $G$) and the set of clients is $E$ (the hyperedges in $G$).  There is a link $(v,e)$ from server $v\in V$ to a client $e\in E$ if $v\in e$.  We note that the degree of the clients is
bounded by $f$ and the degree of the servers is bounded by $\Delta$.

\subsection{Parameters and Variables}
\begin{itemize}[topsep=1pt,itemsep=-.5ex,partopsep=1ex,parsep=1ex]
\item The approximation factor parameter is $\eps\in(0,1]$.
The parameter $\beta$ is defined by 
  $\beta\triangleq\epsilon/(f+\eps)$, where $f$ is the rank of the hypergraph.
\item Each vertex $v$ is assigned a level $\ell(v)$ which is a nonnegative integer. 

\item We denote the dual variables at the end of iteration $i$ by
  $\delta_i(e)$ (see Appendix~\ref{sec:primal dual} for a description
  of the dual edge packing linear program).  The amount by which
  $\delta_i(e)$ is increased in iteration $i$ is denoted by
  $\deal_i(e)$. Namely, $\delta_i(e)=\sum_{j\leq i}\deal_j(e)$.

\item The parameter $\alpha\geq 2$ determines the factor by which deals are multiplied.
    %
We determine its value in the analysis in the following section.

\end{itemize}

\subsection{Notation}
\begin{itemize}
\item We say that an edge $e$ is \emph{covered} by $C$ if
  $e\cap C\neq\emptyset$.
\item Let $E(v)\triangleq \{ e\in E\mid v\in e\}$ denote the set of
  hyperedges that contain $v$.
\item For every vertex $v$, the algorithm maintains a subset
  $E'(v)\subseteq E(v)$ that consists of the uncovered hyperedges in
  $E(v)$ (i.e., $E'(v)=\{e\in E(v)\mid e\cap C=\emptyset\}$).
\end{itemize}
\subsection{Algorithm {\sc MWHVC}}
\begin{enumerate}

\item Initialization. Set $C\gets \emptyset$. For every vertex $v$, set level $\ell(v)\gets 0$ and uncovered edges $E'(v)\gets E(v)$. 
\item Iteration $i=0$. Every edge $e$ collects the  weight $w(v)$  and degree $|E(v)|$
  from every vertex $v\in e$, and sets:
  $\deal(e)=0.5 \cdot \min_{v\in e}\{ w(v)/|E(v)|$\}. The value
  $\deal(e)$ is sent to every $v\in e$. The dual variable is set to
  $\delta(e)\gets \deal(e)$.
\item For $i=1$ to $\infty$ do:
  \begin{enumerate}
   \item\label{item:tight} Check $\beta$-tightness. For every $v\not\in C$, if $\sum_{e\in E(v)} \delta(e) \ge (1-\beta)w(v)$, then 
$v$ joins the cover $C$, sends a message to every $e\in E'(v)$ that $e$ is covered, and vertex $v$ terminates.

  \item For every uncovered edge $e$, if $e$ receives a message that it is
    covered, then it tells all its vertices that $e$ is covered and terminates.
      \item For every vertex $v\notin C$, if it receives a message from $e$ that
    $e$ is covered, then $E'(v)\gets E'(v)\setminus \{e\}$.
    If $E'(v)=\emptyset$, then $v$ terminates  (without joining the cover).
 \item\label{item:increment level} Increment levels and scale deals.
  \\
For every active (that has not terminated) vertex: 
\\\hspace*{1cm} While $\sum_{e\in E(v)} \delta(e) > w(v)(1-0.5^{\ell(v)+1})$ do 
        \begin{enumerate}[leftmargin=4\parindent]
      \item $\ell(v)\gets \ell(v)+1$
      \item For every $e\in E'(v)$: $\deal(e)\gets 0.5\cdot \deal(e)$
  \end{enumerate}    

  \item \label{item:good}
    For every active vertex, if
    $\sum_{e\in E'(v)} \deal(e) \leq \frac{1}{\alpha}\cdot 0.5^{\ell(v)+1} \cdot
    w(v)$, then send the message ``raise'' to every $e\in E'(v)$; otherwise,
    send the message ``stuck'' to every $e\in E'(v)$.
  \item \label{item:update}
    For every uncovered edge $e$, 
    if \emph{all} incoming messages are ``raise''
    $\deal(e)\gets\alpha\cdot \deal(e)$. Send $\deal(e)$ to
    every $v\in e$, who updates $\delta(e)\gets \delta(e)+\deal(e)$.
  \end{enumerate}
\end{enumerate}

\paragraph*{Termination.}
\quad{}Every vertex $v$ terminates when either $v\in C$ or every edge $e\in E(v)$ is covered (i.e., $E'(v)=\emptyset$).
Every edge $e$ terminates when it is covered (i.e., $e\cap C\neq\emptyset$).  
We say that the algorithm has terminated if all the vertices and edges have terminated.
\paragraph*{Execution in \congest.}
\quad{}See Section~\ref{sec:congest} in the Appendix for a discussion of how
Algorithm \WHVC is executed in the \congest model.

\section{Algorithm Analysis}
In this section, we analyze the approximation ratio and the running time of the algorithm. Throughput the analysis, we attach an index $i$ to the variables $\deal_i(e), \delta_i(e)$ and $\ell_i(v)$.
The indexed variable refers to its value at the end of the $i$'th iteration.
\subsection{Feasibility and Approximation Ratio}
The following invariants are satisfied throughout the execution of the algorithm.
In the following claim we bound the sum of the deals
of edges incident to a vertex.
\begin{claim}\label{claim:vault}
  If $v\not\in C$ at the end of iteration $i$, then $\sum_{e\in E'(v)} \deal_{i} (e) \leq  0.5^{\ell_i(v)+1}\cdot w(v)$.
\end{claim}
\begin{proof}
  The proof is by induction on $i$.
  For $i=0$, the claim holds because  $\ell_0(v)=0$ and $\deal_0(e)\leq 0.5\cdot w(v)/|E(v)|$.
   The induction step, for $i\geq 1$, considers two cases.  
(A)~If $v$ sends a ``raise'' message in iteration $i$, then 
  Step~\ref{item:good} implies that
  $\sum_{e\in E'(v)} \deal_{i}(e)\leq 0.5^{\ell(v)+1}\cdot w(v)$, as required.
 (B)~ Suppose $v$ sends a ``stuck'' message in iteration $i$.  
 By Step~\ref{item:increment level}, $\deal_i(e)\leq 0.5^{\ell_i(e)-\ell_{i-1}(e)}\cdot \deal_{i-1}(e)$ for every $e\in E'(v)$.
 The induction hypothesis states that $\sum_{e\in E'(v)} \deal_{i-1} (e) \leq 0.5^{\ell_{i-1}(v)+1}\cdot w(v)$. The claim follows by combining these inequalities.
 \vspace{-3mm}
\end{proof}

\medskip\noindent If an edge $e$ is covered in iteration $j$, then
$e$ terminates and $\delta_i(e)$ is not set for $i\geq j$. In this case, we define $\delta_i(e)=\delta_{j-1}(e)$, namely, the last value assigned to a dual variable. 
\begin{claim}\label{claim:feasible}
For every $i\geq 1$ and
every vertex $v\notin C$ the following inequality holds:
  \begin{align}\label{eq:sandwich}
      w(v)(1-0.5^{\ell_i(v)}) \leq \sum_{e\in E(v)}\delta_{i-1}(e)\le(1-0.5^{\ell_i(v)+1})\cdot w(v) \;.
  \end{align}
  In addition
  the dual variables $\delta_i(e)$ constitute a feasible edge packing.~Namely,
  \begin{align*}
    \sum_{e\in E(v)} \delta_i(e) &\leq w(v) &\text{for every vertex $v\in V$,}\\
    \delta_i(e)&\geq 0&\text{for every edge $e\in E$.}
  \end{align*}
  
\end{claim}
\begin{proof}
We prove the claim by induction on the iteration number $i$.
To simplify the proof, we reformulate the statement of the feasibility of the dual variables to $i-1$.
We first prove the induction basis for $i=1$. 

\paragraph{Proof of Eq.~\ref{eq:sandwich} for $i=1$.}~
Fix a vertex $v$.
At the end of iteration $0$, $\ell_0(v)=0$ and $0<\deal_0(e)\leq w(v)/(2|E(v)|)$, for every $e\in E(v)$.
Hence $0<\sum_{e\in E(v)} \deal_0(e) \leq w(v)/2$. Because $\delta_0(e)=\deal_0(e)$, the condition in Step~\ref{item:increment level} does not hold, and $\ell_1(v)=\ell_0(v)=0$.  We conclude that Eq.~\ref{eq:sandwich} holds for $i=1$. 
\paragraph{Proof of feasibility of $\delta_0(e)$ for $i=1$.}~
Non-negativity follows from the fact that $\delta_0(e)=\deal_0(e)>0$.
The packing constraint for vertex $v$ is satisfied because $\sum_{e\in E(v)}\deal_0(e) \leq w(v)/2$.
This completes the proof of the induction basis. 

\medskip\noindent
We now prove the induction step assuming that Eq.~\ref{eq:sandwich} holds for $i-1$.
\paragraph{Proof of Eq.~\ref{eq:sandwich} for $i>1$.}~
Since $v$ is not in the cover it is also not $\beta$-tight.
Step~\ref{item:increment level} in iteration $i$ increases $\ell(v)$ until Eq.~\ref{eq:sandwich} holds for $i$.

\paragraph{Proof of feasibility of $\delta_{i-1}(e)$ for $i>1$.}~
Consider a vertex $v$. If $v$ joins $C$ in iteration $i-1$, then $\delta_{i-1}=\delta_{i-2}$, and the packing constraint of $v$ holds by the induction hypothesis. If $v\notin C$, then by $\delta_{i-1}(e)=\delta_{i-2}(e)+\deal_{i-1}(e)$, Claim~\ref{claim:vault}, and the induction hypothesis for Eq.~\ref{eq:sandwich}, we have
\begin{align*}
    \sum_{e\in E(v)} \delta_{i-1}(e) &= \sum_{e\in E(v)} \parentheses{\delta_{i-2}(e)+\textstyle{\deal_{i-1}(v)}}
    &\leq \parentheses{1-0.5^{\ell_{i-1}(v)+1}+0.5^{\ell_{i-1}(v)+1}}\cdot w(v)=w(v)\;.\qquad{}\qedhere
\end{align*}
  \end{proof}

\medskip\noindent  Let $\opt$ denote the cost of an optimal (fractional) weighted
  vertex cover of $G$.
  \begin{corollary}
Upon termination, the approximation ratio of Algorithm \WHVC is $f+\eps$.
\end{corollary}
\begin{proof}
  Throughout the algorithm, the set $C$ consists of $\beta$-tight
  vertices. By Claim~\ref{claim:ratio PCWHVC}, $w(C)\leq (f+\eps)\cdot \opt$. Upon
  termination, $C$ constitutes a vertex cover, and the corollary
  follows.
\end{proof}

\subsection{Communication Rounds Analysis}\label{sec:runtime}
In this section, we prove that the number of communication rounds of
Algorithm~\MWHVC{}\ is bounded~by (where $\gamma>0$ is a constant, e.g., $\gamma=0.001$)
{\small
$$O\parentheses{f\cdot \log(f/\eps) + \frac{\log \Delta}{\gamma\cdot \log \log \Delta}+\min\set{\log\Delta,f\cdot \log(f/\eps)\cdot (\log\Delta)^{\gamma}}}.$$}
It suffices to bound the number of iterations because each iteration consists of a constant number of communication rounds.

\medskip
\noindent
Let $z\triangleq \ceil{\log_2 {\frac{1}{\beta}}}$. Note that $z=O\parentheses{\log ({f}/{\eps})}$.
\begin{claim}\label{claim:levelLessThanZ}
The level of every vertex is always less than $z$.
\end{claim} 
\begin{proof}
Assume that $\ell(v)\geq z$.
By Eq.~\ref{eq:sandwich}, $\sum_{e\in E(v)} \delta_{i-1}(e) \geq w(v)\cdot (1-2^{-z})\geq (1-\beta)\cdot w(v)$.
This implies that $v$ is $\beta$-tight and joins the cover in Line~\ref{item:tight} before $\ell(v)$ reaches $z$.
\end{proof}

\subsubsection{Raise or Stuck Iterations}

\begin{definition}[$e$-raise and $v$-stuck iterations]
  An iteration $i\geq 1$ is an $e$-\emph{raise} iteration if in Line~\ref{item:update} we multiplied $\deal(e)$ by $\alpha$.
  An iteration $i\geq 1$ is a $v$-\emph{stuck} iteration if $v$ sent the message ``stuck'' in iteration $i$.
\end{definition}
Note that if iteration $i$ is a $v$-stuck iteration and $v\in e$, then $\deal_i(e)\le \deal_{i-1}(e)$ and $i$ is not an $e$-raise iteration.
\medskip\noindent
We bound the number of $e$-raise iterations as~follows.
\begin{lemma}\label{lemma:regular passes}
The number of $e$-raise iterations is bounded by $\log_\alpha ( \Delta \cdot 2^{f\cdot z})$.
\end{lemma}
\begin{proof}
  Let $v^*$ denote a vertex with minimum normalized weight in $e$ (i.e., $v^*\in \text{argmin}_{v\in e}\{ w(v)/|E(v)|$\}).
  The first deal satisfies
  $\deal_0(e)= 0.5\cdot w(v^*)/|E(v^*)| \geq 0.5\cdot
  w(v^*)/\Delta$.  By Claim~\ref{claim:vault},
   $\deal_i(e)\leq 0.5\cdot w(v^*)$. The deal is multiplied by
  $\alpha$ in each $e$-raise iteration and is halved at most $f \cdot z$ times. The bound on the number of halvings holds because the number of vertices in the edge is bounded by $f$, and each vertex halves the deal each time its level is incremented. The lemma follows.
\end{proof}
\medskip\noindent
We bound the number of $v$-stuck iterations as~follows.
\begin{lemma}\label{lemma:good}
For every vertex $v$ and level $\ell(v)$, the number of $v$-stuck iterations is bounded by $\alpha$.
\end{lemma}
\begin{proof}
Notice that when $v$ reached the level $\ell(v)$, we had $\sum_{e\in E(v)} \delta(e) \ge w(v)(1-0.5^{\ell(v)})$.
%
The number of $v$-stuck iterations is then bounded by the number of times it can send a "stuck" message without reaching $\sum_{e\in E(v)} \delta(e) > w(v)(1-0.5^{\ell(v)+1})$. Indeed, once this inequality holds, the level of $v$ is incremented.
Every stuck iteration implies, by Line~\ref{item:good}, that $\sum_{e\in E'(v)} \deal(e) > \frac{1}{\alpha}\cdot 0.5^{\ell(v)+1} \cdot w(v)$. 
Therefore, we can bound the number of iteration by 
$
    \frac{w(v)(1-0.5^{\ell(v)+1}) - w(v)(1-0.5^{\ell(v)})}{  \frac{1}{\alpha}\cdot 0.5^{\ell(v)+1} \cdot w(v)} = \alpha.
$\qedhere

\end{proof}
\medskip\noindent

\subsubsection{Putting it Together}
\begin{theorem}\label{thm:iterations}
For every $\alpha\geq 2$, the number of iterations of Algorithm~$\alg$
is
{\small $$O\parentheses{\log_\alpha \Delta + f\cdot z\cdot \alpha}=O\parentheses{\log_\alpha \Delta + f\cdot \log(f/\eps)\cdot \alpha}.$$}
\end{theorem}
\begin{proof}
Fix an edge $e$. We bound the number of iterations until $e$ is covered. Every iteration is either an $e$-raise iteration or a $v$-stuck iteration for some $v\in e$. 
Since $e$ contains at most $f$ vertices, we conclude that the number of iterations is bounded by the number of $e$-stuck iterations plus the sum over $v\in e$ of the number of $v$-stuck iterations.  The theorem follows from Lemmas~\ref{lemma:regular passes} and~\ref{lemma:good}.
\end{proof}

In Theorem~\ref{thm:running time} we assume that all the vertices know the maximum degree $\Delta$ and that $\Delta\ge 3$.  The assumption that the maximal degree $\Delta$ is known to all vertices is not required. Instead, each hyperedge $e$
can compute a local maximum degree $\Delta(e)$, where
$\Delta(e)\triangleq\max_{u\in e}|E(u)|$.  The local maximum degree
$\Delta(e)$ can be used instead of $\Delta$ to define local value
of the multiplier $\alpha=\alpha(e)$.
Let $T(f,\Delta,\eps)$ denote the round complexity of Algorithm \WHVC.
By setting $\alpha$ appropriately, we bound the running time as follows.
\begin{theorem}\label{thm:running time}\footnote{The statement of the theorem is asymptotic. This means that for every constant $\gamma$, it holds that either $\log^{\gamma/2} \Delta \geq \log \log \Delta$ or $\Delta$ is bounded by constant (determined by $\gamma$)  in which case expressions involving $\Delta$ can be omitted from the asymptotic expression. }
Let $\gamma>0$ denote a constant 
and 
{\small
\begin{align*}
    \alpha&=
\begin{cases}
\max\parentheses{2,\frac{\log \Delta}{f\cdot \log(f/\eps)\cdot \log \log \Delta}}&\mbox{if $\frac{\log \Delta}{f\cdot \log(f/\eps)\cdot \log \log \Delta}\ge (\log\Delta)^{\gamma/2}$}\\
2&\mbox{Otherwise.}
\end{cases}
\end{align*}
}
Then,  the round complexity of Algorithm \WHVC satisfies:
 {\small 
 $$T(f,\Delta,\eps) = O\parentheses{f\cdot \log(f/\eps) + \frac{\log \Delta}{ 
 \log \log \Delta}+\min\set{\log\Delta,f\cdot \log(f/\eps)\cdot (\log\Delta)^{\gamma}}}.$$}
\end{theorem}
\begin{proof}
First, consider the case where 
$\alpha=\frac{\log \Delta}{f\cdot \log(f/\eps)\cdot \log \log \Delta}\geq 2$
and 
$\alpha \ge (\log\Delta)^{\gamma/2}$. 
This means that the runtime is bounded by
$O\parentheses{\log_\alpha \Delta + f\cdot \log(f/\eps)\cdot \alpha} =
O\parentheses{\frac{\log \Delta}{ \gamma\cdot \log \log \Delta}}\;.$

Second, assume that $\alpha=2$ and $2>\frac{\log \Delta}{f\cdot \log(f/\eps)\cdot \log \log \Delta}\geq (\log \Delta)^{\gamma/2}$.
Then,  the round complexity of Algorithm \WHVC is bounded by
$O\parentheses{\log_2 \Delta + f\cdot \log(f/\eps)}
=O\parentheses{f\cdot \log(f/\eps)\cdot\log\log\Delta}
=O\parentheses{f\cdot \log(f/\eps)},
\;$ where the last transition is correct as $2> (\log \Delta)^{\gamma/2}$ implies that $\Delta$ is constant.

Finally, assume that $\frac{\log \Delta}{f\cdot \log(f/\eps)\cdot \log \log \Delta}< (\log\Delta)^{\gamma/2}$ which implies that 
$$\log\Delta \le \min\set{\log\Delta,f\cdot \log(f/\eps)\cdot (\log\Delta)^{\gamma/2}\cdot \log\log\Delta}=O\parentheses{\min\set{\log\Delta,f\cdot \log(f/\eps)\cdot (\log\Delta)^{\gamma}}}.$$
Therefore, since $\alpha=2$ in this case, the runtime is bounded by 
$$O\parentheses{\log\Delta + f\cdot \log(f/\eps)} = 
O\parentheses{f\cdot \log(f/\eps) + \min\set{\log\Delta,f\cdot \log(f/\eps)\cdot (\log\Delta)^{\gamma}}}.\qedhere$$
\end{proof}
Let $W\triangleq \max_v w(v) / \min_v w(v)$.
By setting $\eps=1/(nW)$, we conclude the following result for an $f$-approximation (recall that we assume that vertex degrees and weights are polynomial in $n$):
\begin{corollary}
Algorithm \MWHVC computes an $f$-approximation 
in $O(f\log n)$ rounds.
\end{corollary}
\medskip\noindent
Additionally, we get the following range of parameters for which the round complexity is still optimal:
\begin{corollary}
Let $f=O\parentheses{(\log \Delta)^{0.99}}$ and $\eps=(\log\Delta)^{-O(1)}$. Then, Algorithm \MWHVC computes an $(f+\eps)$-approximation in $O\parentheses{\frac{\log \Delta}{ \log \log \Delta}}$ rounds.
\end{corollary}
For $f=O(1)$ we also get an extension of range of parameters for which the round complexity is optimal. This extension is  almost exponential compared to the allowed $\eps=(\log\Delta)^{-O(1)}$ in~\cite{BEKS18}. 
\begin{corollary}
Let $f=O(1)$ and $\eps=2^{-O\parentheses{(\log \Delta)^{0.99}}}$. Then our algorithm computes an $(f+\eps)$-approximation and terminates in $O\parentheses{\frac{\log \Delta}{ \log \log \Delta}}$ rounds.
\end{corollary}

\section{Approximation of Covering ILPs}
In this section, we present a reduction from solving covering integer linear programs (ILPs) to \MWHVC. This reduction implies that one can distributively compute approximate solutions to covering ILPs using a distributed algorithm for \MWHVC.

\paragraph{Notation.} Let $\NN$ denote the set of natural 
numbers.
Let $A$ denote a real $m\times n$ matrix, $\vec{b}\in \RR^n$, and $\vec{w}\in\RR^n$.
Let $LP(A,\vec{b},\vec{w})$ denote the linear program $\min~\vec{w}^{T} \cdot \vec{x}$ subject to $A\cdot \vec{x} \geq \vec{b}$ and $\vec{x} \geq 0$.
Let $ILP(A,\vec{b},\vec{w})$ denote the integer linear program $\min~\vec{w}^{T} \cdot \vec{x}$ subject to $A\cdot \vec{x} \geq \vec{b}$ and $\vec{x} \in\NN^n$.

\begin{definition}
The linear program $LP(A,\vec{b},\vec{w})$ 
and integer linear program $ILP(A,\vec{b},\vec{w})$ are \emph{covering programs} if all the components in $A,\vec{b}, \vec{w}$ are non-negative. 
\end{definition}

\subsection{Distributed Setting}

We denote the number of rows of the matrix $A$ by $m$ and the number of columns by $n$. Let $f(A)$ (resp., $\Delta(A)$) denote the maximum number of nonzero entries in a row (resp., column) of $A$.

Given a covering ILP, $ILP(A,\vec{b},\vec{w})$, the communication network $N(ILP)$ over which the ILP is solved is the bipartite graph $N=(X\times C, E)$, where:
$X=\{x_j\}_{j\in [n]}$, $C=\{c_i\}_{i\in [m]}$, and $E=\{(x_j,c_i)\mid A_{i,j}\neq 0\}$.
We refer to the nodes in $X$ as variable nodes, and to those in $C$ as constraint nodes.  Note that the maximum degree of a constraint node is $f(A)$ and the maximum degree of a variable node is  $\Delta(A)$.

We assume that the local input of every variable node $x_j$ consists of $w_j$ and the $j$'th column of $A$ (i.e., $A_{i,j}$, for $i\in [m]$). Every constraint vertex $c_i$  is given the value of $b_i$ as its local input. We assume that these inputs can be represented by $O(\log (nm))$ bits. %
In $(f+1)$ rounds, every variable node $v_j$ can learn all the components in the rows $i$ of $A$ such that $A_{i,j}\neq 0$ \mbox{as well as the component $b_i$. }

\subsection{Zero-One Covering Programs}
The special case in which a variable may be assigned only the values $0$ or $1$ is called a \emph{zero-one program}. We denote the zero-one covering ILP induced by a matrix $A$ and vectors $\vec{b}$ and $\vec{w}$ by $ZO(A,\vec{b},\vec{z})$.
Every instance of the \MWHVC\ problem is a zero-one program in which the matrix $A$ is the incidence matrix of the hypergraph. The following lemma deals with the converse reduction.

\begin{lemma}\label{lemma:reduce ZO HVC}
Every feasible zero-one covering program  $ZO(A,\vec{b},\vec{w})$ 
can be reduced to an \MWHVC instance with rank $f'<f(A)$ and degree $\Delta'< 2^{f(A)}\cdot \Delta(A)$. \end{lemma}
\begin{proof}
Let $A_i$ denote the $i$'th row of the matrix $A$. For a subset $S\subseteq [n]$, let $I_S$ denote the indicator vector of $S$. Let $\vec{x}\in \{0,1\}^n$ and let $\sigma_i\triangleq \{j\in [n] \mid A_{i,j}\neq 0\}$.
Feasibility implies that $A_i\cdot I_{\sigma_i} \geq b_i$, for every row $i$.
Let $\mathcal{S}_i$ denote the set of all subsets $S\subset [n]$ such that  the indicator vector $I_S$ does not satisfy the $i$'th constraint, i.e., $A_i \cdot I_S < b_i$
%
The $i$'th constraint is not satisfied, i.e., $A_i\cdot \vec{x}<b_i$,  if and if only there exists a set $S\in \mathcal{S}_i$ such that $\vec{x}=I_S$. Hence, $A_i\cdot \vec{x} < \vec{b}$ if and only if the truth value of the following DNF formula is false: $\varphi_i(x) \triangleq\bigvee_{S\in\mathcal{S}_i}~ \bigwedge_{j\in \sigma_i\setminus S} \text{not}({x}_j)$.
By De Morgan's law, we obtain that $\text{not}(\varphi_i(x))$ is equivalent to a monotone CNF formula $\psi_i(x)$ such that $\psi_i(x)$ has less than $2^{f(A)}$ clauses, each of which has length less than $f(A)$.
%
We now construct the hypergraph $H$ for the \MWHVC instance as follows. For every row $i$ and every $S\in\mathcal{S}$, add the hyperedge $e_{i,S}=\sigma_i\setminus S$. (Feasibility implies that $e_{i,S}$ is not empty.) Given a vertex cover $C$ of the hypergraph, every hyperedge is stabbed by $C$, and hence $I_C$ satisfies all the  formulae $\psi_i(x)$, where $i\in[m]$. Hence, $A_i\cdot I_C \geq b_i$, for every $i$. The converse direction holds as well, and the lemma follows.
\end{proof}

How does 
$N(ILP)$ (a bipartite graph with $m+n$ vertices) simulate the execution of \MWHVC over the hypergraph $H$? 
Each variable node $x_j$ simulates all hyperedges $e_{i,S}$, where $j\in\sigma_i$ and $S\in\mathcal{S}_i$.
First, the variable nodes exchange their weights with the variables nodes they share a constraint with in $O(f(A))$ rounds -- first every $x_j$ broadcasts its own weight and then each $c_i$ sends all neighbors the weights it received.
At each iteration, the variable node sends a raise/stuck message and whether its level was incremented. 
\ifdefined\fullVersion
\footnote{One needs to modify the \MWHVC algorithm slightly so that, in each iteration, the level of every vertex is increased by at most $1$. We defer the details to Appendix~\ref{app:oneLevelup}.} 
\else
\footnote{One needs to modify the \MWHVC algorithm slightly so that, in each iteration, the level of every vertex is increased by at most $1$. We defer the details to the full version.} 
\fi
Notice that the number of rounds required for each such iteration is $O(1+f(A)/\log n)$ (i.e., constant for $f(A)=O(\log n)$).
Each edge node $c_i$ then broadcasts to all vertices two $f(A)$-bit messages that indicate two subsets of vertices of the edge: those that sent a raise message (the complement sent a stuck message) and those that incremented their level.
Each variable node $x_j$ knows how to update its deal with every $e_{i,S}$ for which $j\in\sigma_i$ and $S\in\mathcal{S}_i$.

\medskip\noindent
We 
summarize the complexity of the distributed algorithm \mbox{for solving a zero-one covering 
program.}
\begin{claim}\label{claim:ZO alg}
  There exists a distributed \congest algorithm for computing an $(f+\eps)$-approximate solution for zero-one covering programs with running time of 
  $O({(1+f(A)/\log n)\cdot} T(f(A), 2^{f(A)}\cdot \Delta(A),\eps))$, where $T(f,\Delta,\eps)$ is the running time of Algorithm \MWHVC.
\end{claim}
%

\subsection{Reduction of Covering ILPs to Zero-One Covering}
Consider the covering ILP $ILP(A,\vec{b},\vec{w})$. We present a reduction of the ILP to a \mbox{zero-one covering program.}

\begin{definition}
Define $M(A,\vec{b})\triangleq \max_j \max_i \set{\frac{b_i}{A_{i,j}}\mid A_{i,j}\neq 0}$.
\end{definition}
\noindent
We abbreviate and write $M$ for $M(A,\vec{b})$ when the context is clear.
\begin{proposition}
Limiting $\vec{x}$ to the box $[0,M]^{n}$ does not increase the optimal value of the ILP. 
\end{proposition}

\begin{claim}\label{claim:ILP to ZO}
  Every covering ILP $ILP(A,\vec{b},\vec{w})$ can be solved by a zero-one covering program  $ZO(A',\vec{b},\vec{w'})$, where $f(A')\leq f(A)\cdot \ceil{\log_2(M)+1}$ and $\Delta(A')=\Delta(A)$.
  \end{claim}
\begin{proof}
Let 
$B = \ceil{\log_2 M}$.
Limiting each variable $x_j$ by $M$ means that we can replace $x_j$ by $B$ zero-one variables $\set{x_{j,\ell}}_{\ell=0}^{B-1}$ that correspond to the binary representation of $x_j$, i.e., $x_j=\sum_{\ell=0}^{B-1} 2^{\ell}\cdot x_{j,\ell}$.  This replacement means that the dimensions of the matrix $A'$ are $m\times n'$, where $n'=n\cdot B$. The $j$'th column $A^{(j)}$ of $A$ is replaced by $B$ columns, indexed $0$ to $B-1$, where the $\ell$'th column equals $2^{\ell}\cdot A^{(j)}$. The vector $\vec{w'}$ is obtained by duplicating and scaling the entries of $\vec{w}$ in the same fashion.
\end{proof}

\noindent
Combining Claims~\ref{claim:ZO alg} and~\ref{claim:ILP to ZO}, we obtain the following result.
\begin{claim}
  There exists a distributed \congest algorithm for computing an $(f+\eps)$-approximate solution for covering integer linear programs $ILP(A,\vec{b},\vec{w})$ with running time of $O({(1+f(A)/\log n)\cdot} T(f(A)\cdot \log M, 2^{f(A)}\cdot M\cdot \Delta(A),\eps))$, where $T(f,\Delta,\eps)$ is the running time of Algorithm \MWHVC.
\end{claim} 
\begin{proof}
Let $f_{HVC}, f_{ZO}, F_{ILP}$ denote the ranks of the hypergraph vertex cover instance, zero-one covering program, and ILP, respectively. We use the same notation for maximum degrees.
The reduction of zero-one programs to \MWHVC in Lemma~\ref{lemma:reduce ZO HVC} implies that $f_{HVC}\leq f_{ZO}$ and $\Delta_{HVC}\leq 2^{f_{ZO}}\cdot \Delta_{ZO}$.
The reduction of covering ILPs to zero-one programs in Claim~\ref{claim:ILP to ZO} implies that $f_{ZO}\leq f_{ILP}\cdot (1+\log M)$ and $\Delta_{ZO}\leq \Delta_{ILP}$.
The composition of the reductions gives $f_{HVC}\leq f_{ILP}\cdot (1+\log M)$ and $\Delta_{HVC}\leq 2^{f_{ILP}\cdot (1+\log M)}\cdot \Delta_{ILP}=2^{F_{ILP}}\cdot 2M\cdot \Delta_{ILP}$ 
\end{proof}

After some simplifications, the running time of the resulting algorithm for $(f+\eps)$-approximate integer covering linear programs is 
$\ILPcomplexity$.  
\section*{Acknowledgements}
We thank the anonymous reviewers for their helpful remarks.

\bibliographystyle{alpha}
\bibliography{paper}
\appendix
\section{Primal-Dual Approach}
 \label{sec:primal dual}

\medskip
\noindent
The fractional LP relaxation of \WHVC\ is defined as follows.
\begin{alignat*}{2}\label{eq:P}
  & \text{minimize: }  \sum_{v\in V}w(v) \cdot x(v) &&\\
  & \text{subject to: }& \quad & \\
&\tag{$\mathcal{P}$}
  \begin{aligned}
                   \sum_{v\in e} x(v)  &\geq 1,& \forall e\in E\\
                   x(v) &\geq 0,&\forall v\in V
                 \end{aligned}
\end{alignat*}
The dual LP is an \emph{Edge Packing} problem defined as follows:
\begin{alignat*}{2}\label{eq:D}
  & \text{maximize: }  \sum_{e\in E} \delta(e)&&\\
  & \text{subject to: }& \quad &\\
&\tag{$\mathcal{D}$}
  \begin{aligned}
                   \sum_{e\ni v} \delta(e) &\leq w(v),& \forall v\in V\\
                   \delta(e) &\geq 0,&\forall e\in E
                 \end{aligned}
\end{alignat*}
The following claim is used for proving the approximation ratio of
the \WHVC algorithm.
\begin{claim}\label{claim:ratio PCWHVC}
  Let $\opt$ denote the value of an optimal fractional solution of the
  primal LP~\eqref{eq:P}. Let $\{\delta(e)\}_{e\in E}$ denote a
  feasible solution of the dual LP~\eqref{eq:D}.  Let $\eps\in(0,1)$
  and $\beta\triangleq \eps/(f+\eps)$. Define the $\beta$-tight 
 vertices by:
  \begin{align*}
    T_{\eps} &\triangleq \{ v\in V \mid \sum_{e\ni v} \delta(e)\geq (1-\beta)\cdot w(v)\}.
  \end{align*}
Then $w(T_{\eps})\leq (f+\eps) \cdot \opt$.
\end{claim}
\begin{proof}
  \begin{align*}
    w(T_{\eps}) &=\sum_{v\in T_{\eps}} w(v)\\
                              &\leq \frac{1}{1-\beta} \cdot \parentheses{\sum_{v\in T_{\eps}} \sum_{e\ni v} \delta(e)}\\
                              &\leq \frac{f}{1-\beta} \sum_{e\in E} \delta(e) \leq (f+\eps)\cdot \opt.
  \end{align*}
  The last transition follows from $f/(1-\beta)=f+\eps$ and by  weak 
  duality. The claim follows.
\end{proof}
\section{Adaptation to the CONGEST model}\label{sec:congest}
To complete the discussion, we need to show that the message lengths
in Algorithm \WHVC are $O(\log n)$.
\begin{enumerate}
\item In round $0$, every vertex $v$ sends its weight $w(v)$ and
  degree $|E(v)|$ to every hyperedge in $e\in E(v)$. We assume that
  the weights and degrees are polynomial in $n$, hence the length of the binary
  representations of $w(v)$ and $|E(v)|$ is $O(\log n)$.
  
  Every hyperedge $e$ sends back to every $v\in e$ the pair $(w(v_e),|E(v_e)|)$,
  where $v_e$ has the smallest normalized weight, i.e.,
  $v_e=\argmin_{v\in e} \{w(v)/|E(v)|\}$.

  Every vertex $v\in e$ locally computes $\deal_0(e)=\beta \cdot w(v_e)/|E(v_e)|$
  and $\delta_0(e)=\deal_0(e)$.

\item In round $i\geq 1$, every vertex sends messages. Messages of the sort:
  ``$e$ is covered'', ``raise'', or ``stuck'' require
  only a constant number of bits.  
  The increment of $v$'s level needs to be sent to the edges in $E(v)$. These increments require $O(\log z) =O(\log n)$ bits.

\item Every edge $e$ sends to every $v\in e$ the number of times that $\deal(e)$ is halved in this iteration. This message is $O(\log z)$ bits long.

\item Every edge sends the final deal to the vertices.
Instead of sending the value of the deal, the edge can
send a single bit indicating whether the deal was multiplied by $\alpha$.

\item Finally, if $\alpha=\alpha(e)$ is set locally based on the local
  maximum degree $\max_{v\in e} |E(v)|$, then every vertex $v$ sends
  its degree to all the edges $e\in E(v)$. The local maximum degree
  for $e$ is sent to every vertex $v\in V$, and this parameter is used
  to compute $\alpha(e)$ locally.
\end{enumerate}
\ifdefined\fullVersion
\section{Algorithm with At Most One Level-up per Iteration}\label{app:oneLevelup}
\fi
We propose to do a single change that will ensure that no vertex levels up more than once per iteration. To that end, we modify Line~\ref{item:update} of our \MWHVC algorithm to
\begin{itemize}
    \item 
    For every uncovered edge $e$, 
    if \emph{all} incoming messages are ``raise''
    $\deal(e)\gets\alpha\cdot \deal(e)$. Send $\deal(e)$ to
    every $v\in e$, who updates $\delta(e)\gets \delta(e)+\deal(e)/2$.
\end{itemize}
That is, the algorithm remains intact except that the dual variables $\delta(e)$ are raised by $\deal(e)/2$ rather than by $\deal(e)$.
Intuitively, this guarantees that a vertex's slack does not reduce by more than 50\% in each iteration and therefore its level may increase by at most one.
We now revisit the proof of~\ref{claim:feasible}, and, specifically, the \emph{Proof of feasibility of $\delta_{i-1}(e)$ for $i>1$.}

\paragraph{Proof of feasibility of $\delta_{i-1}(e)$ for $i>1$.}~
Consider a vertex $v$. If $v$ joins $C$ in iteration $i-1$, then $\delta_{i-1}=\delta_{i-2}$, and the packing constraint of $v$ holds by the induction hypothesis. If $v\notin C$, then by $\delta_{i-1}(e)=\delta_{i-2}(e)+\deal_{i-1}(e)/2$, Claim~\ref{claim:vault}, and the induction hypothesis for Eq.~\ref{eq:sandwich}, we have
\vspace{-3mm}
\begin{multline}\label{eq:lvlupCap}
\vspace{-3mm}
    {\sum_{e\in E(v)} \delta_{i-1}(e) }= \sum_{e\in E(v)} \parentheses{\delta_{i-2}(e)+{\textstyle{\deal_{i-1}(v)/2}}}\\
    {\leq \parentheses{1-0.5^{\ell_{i-1}(v)+1}+0.5^{\ell_{i-1}(v)+2}}\cdot w(v)= \parentheses{1-0.5^{(\ell_{i-1}(v)+1)+1}}\cdot w(v)\;.\qquad{}}\qedhere
\end{multline}
{As evident by Eq.~\ref{eq:lvlupCap}, the vertex $v$'s level may increase by at most once in each iteration.}
\begin{corollary}\label{cor:lvlupBound}
{For any iteration $i\ge 1$ and vertex $v$: $\ell_{i}(v)\le\ell_{i-1}(v) + 1$.}
\end{corollary}
We note that the change to the algorithm does not affect the correctness of claims~\ref{claim:vault},~\ref{claim:feasible},~\ref{claim:levelLessThanZ} and Lemma~\ref{lemma:regular passes}. Lemma~\ref{lemma:good} changes slightly, as there can now be twice as many $v$-stuck iterations:
\begin{lemma}\label{lemma:good2}
For every vertex $v$ and level $\ell(v)$, the number of $v$-stuck iterations is bounded by {$2\alpha$}.
\end{lemma}
\begin{proof}
Notice that when $v$ reached the level $\ell(v)$, we had $\sum_{e\in E(v)} \delta(e) \ge w(v)(1-0.5^{\ell(v)})$.
%
The number of $v$-stuck iterations is then bounded by the number of times it can send a "stuck" message without reaching $\sum_{e\in E(v)} \delta(e) > w(v)(1-0.5^{\ell(v)+1})$. Indeed, once this inequality holds, the level of $v$ is incremented.
Every stuck iteration implies, by Line~\ref{item:good}, that $\sum_{e\in E'(v)} \deal(e) > \frac{1}{\alpha}\cdot 0.5^{\ell(v)+1} \cdot w(v)$. 
Therefore, we can bound the number of iteration by 
$
    \frac{w(v)(1-0.5^{\ell(v)+1}) - w(v)(1-0.5^{\ell(v)})}{  {\frac{1}{\alpha}\cdot 0.5^{\ell(v)+1} \cdot w(v)/2}} = {2\alpha}.
$\qedhere
\end{proof}
\medskip\noindent
We conclude that the algorithm remains correct and its asymptotic complexity does not change.
\end{document}